\documentclass{article}
\usepackage{graphicx} 
\usepackage{paithan}
\usepackage{amssymb}
\usepackage{xfrac}

\usepackage{mathrsfs}
\usepackage{mathtools}
\usepackage{amsfonts}
\usepackage{amsthm}

\newtheorem{theorem}{Theorem}[section]
\newtheorem{lemma}[theorem]{Lemma}

\newtheorem{corollary}[theorem]{Corollary}

\theoremstyle{definition}
\newtheorem{definition}[theorem]{Definition}
\newtheorem{problem}[theorem]{Open problem}

\usepackage{xspace}
\usepackage{hyperref}
\usepackage{verbatim}
\usepackage{tikz}
\usetikzlibrary{positioning}

\title{It's Hard to PArcK}
\author{Kyle Burke $^1$, Jeffrey Leman $^1$, Craig Tennenhouse $^2$}
\date{$^1$Florida Southern College, Lakeland, FL 33801, USA\\
$^2$University of New England, Biddeford, ME 04005, USA}

\newcommand{\parck}{\ruleset{PArcK}}
\newcommand{\poscnf}{\ruleset{Positive CNF}}

\begin{document}

\maketitle

\begin{abstract}
    We show that \ruleset{Partizan Arc Kayles} (\ruleset{PArcK}), a generalization of \ruleset{Domineering} to graphs, is \cclass{PSPACE}-complete via a reduction from \poscnf{} and with recently-discovered techniques for creating \parck{} positions with high temperature.  The reduction uses only red and blue edges.
\end{abstract}

\section{Introduction}

Combinatorial Game Theory (CGT) is the study of two-player games with no randomness and no hidden information where players alternate turns.  Under the Normal Play convention, a player loses when there is no move available at the beginning of their turn.  (The last player to move wins.)  \ruleset{Partizan Arc Kayles} (\parck{}) is one such ruleset.  The default player names are Left and Right, but there are conventional other names depending on the game components.

\begin{definition}
    \ruleset{Partizan Arc Kayles} is played on a simple graph where each edge is colored blue, red, or green.  The game is played between Left, who plays on blue edges, and Right, who plays on red.  A turn consists of picking an edge of your color (or green) and removing that edge, both incident points, and any other edges incident to those two points.\footnote{A playable version of the game is available at \url{https://kyleburke.info/DB/combGames/partizanArcKayles.html}.}
\end{definition}

\parck{} is a generalization of the game \ruleset{Domineering}, a Normal Play game played by two players who place dominoes on a grid with each covering two spaces.  Left places dominoes vertically and Right places them horizontally.  Dominoes cannot be placed over other dominoes. \ruleset{Domineering} has been a popular subject of study and can be realized as a special case of \parck{} on a grid where the vertical edges are blue and horizontal edges are red. Thus, the more we learn about \parck{} the better we hope to understand \ruleset{Domineering}.

Normal Play CGT yields many values which can simplify analysis when a game naturally breaks into multiple disjoint pieces and on their turn a player can choose to play on any one of them.  These values can be added together using the disjunctive sum.  Of specific interest in this work are:

\begin{itemize}
    \item $0$, which is a losing position for whichever player goes first.  All positions where neither player has a winning strategy going first are equal to zero.  For example, the position with no edges has value zero because whoever goes first immediately loses.
    \item Positive numbers, which are winning positions for Left, no matter who goes first.  If a position, $G$, has no moves for Right, and Left can take $k$ moves to reach $0$, then $G = k$.
    \item Negative numbers, which are winning positions for Right, no matter who goes first.  For $k>0$, $-k$ is analagous to $k$ but with $k$ moves for Right instead of Left.
    \item A position where all best moves for both players are to zero is known as $*$.
    \item The position $\pm k$ is shorthand for $\{k|-k\}$.
    \item For a game value $g$, the value $g+*$ is abbreviated $g*$.
\end{itemize}

Many other analysis techniques and uses of terminology appear in this paper that are common to CGT.  Interested readers are recommended to check out the excellent CGT textbooks of \cite{WinningWays:2001, LessonsInPlay:2007, SiegelCGT:2013}, especially for details used in the analysis of section \ref{sec:values}.

Algorithmic Combinatorial Game Theory (ACGT), coined in \cite{AlgGameTheoryGONC3}, is the application of algorithmic concepts in the study of combinatorial games.  Results in this area generally show algorithms to solve positions in a specific ruleset in polynomial time or show intractibility of determining which player has a winning strategy on positions of a ruleset.  Computational complexity is detailed in many textbooks such as \cite{PapadimitriouBook:1994}.  This paper proves an intractibility result on \parck{}, which is an improvment on the \cclass{NP}-hardness result of \cite{HANAKA2026103716}.  Recently, the hardness of \parck{} under mis\`ere play was shown in \cite{2511.21888}.

\section{Values of Positions}
\label{sec:values}

In this section, we describe the values of some \parck{} positions that we will use throughout this paper.  

We start with three positions with basic values (birthdays of $1$): $1$, $-1$, and $*$, shown in figure \ref{fig:birthdayOne}.

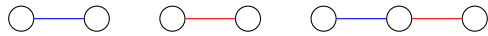
\begin{figure}[h!]
\begin{center}\begin{tikzpicture}[scale=1]
    \node[draw, circle] (blue1) {} ;
    \node[draw, circle] (blue2) [right of=blue1] {} ;
    \node[draw, circle] (red1) [right of=blue2] {};
    \node[draw, circle] (red2) [right of=red1] {};
    \node[draw, circle] (star1) [right of=red2] {};
    \node[draw, circle] (star2) [right of=star1] {};
    \node[draw, circle] (star3) [right of=star2] {};

    \path[-]
        (red1) edge [color=red] (red2)
        (star2) edge [color=red] (star3)
        (blue1) edge [color=blue] (blue2)
        (star1) edge [color=blue] (star2)
    ;
\end{tikzpicture} \end{center}
\caption{Three basic positions.  From left to right: $1$, $-1$, and $*$.}
\label{fig:birthdayOne}
\end{figure}

For our variable gadgets, we will use $-A_2$, shown in figure \ref{fig:nA2} and discovered as part of a more general recursive construction by \cite{craigNeilSvenja:2026} as part of a family of high temperature positions.  On its own, $-A_2$ has value $0$, but if a different edge that includes $m$ is chosen by a player, the two red edges disappear and the value is 2. 

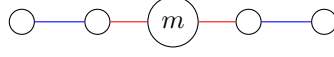
\begin{figure}[h!]
\begin{center}\begin{tikzpicture}[scale=1]
    \node[draw, circle] (m) {$m$} ;
    \node[draw, circle] (a3) [left of=m] {};
    \node[draw, circle] (a4) [right of=m] {};
    \node[draw, circle] (b3) [left of=a3] {};
    \node[draw, circle] (b4) [right of=a4] {};

    \path[-]
        (m) edge [color=red] (a3)
        (m) edge [color=red] (a4)
        (a3) edge [color=blue] (b3)
        (a4) edge [color=blue] (b4)
    ;
\end{tikzpicture} \end{center}
\caption{$-A_2$.  $m$ is the vertex that connects to other parts of the graph.}
\label{fig:nA2}
\end{figure}

We will also use the 3-edge path shown in figure \ref{fig:3path}.

\begin{figure}[h!]
\begin{center}\begin{tikzpicture}[scale=1]
    \node[draw, circle] (y) {} ;
    \node[draw, circle] (x) [right of=y] {};
    \node[draw, circle] (w) [right of=x] {};
    \node[draw, circle] (v) [right of=w] {};

    \path[-]
        (y) edge [color=red] node[left] {} (x)
        (x) edge [color=blue] node[below] {} (w)
        (w) edge [color=red] node[right] {} (v)
    ;
\end{tikzpicture} \end{center}
\caption{The value of this position is $\gameSet{0}{-1} = -\sfrac{1}{2} \pm \sfrac{1}{2}$.}
\label{fig:3path}
\end{figure}
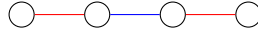

Adding a blue edge to the path in figure \ref{fig:3path} to give it another tail as in figure \ref{fig:2pathFork} doesn't change the value.

\begin{figure}[h!]
\begin{center}\begin{tikzpicture}[scale=1]
    \node[draw, circle] (y) {} ;
    \node[draw, circle] (z) [above of=y] {};
    \node[draw, circle] (x) [right of=y] {};
    \node[draw, circle] (w) [right of=x] {};
    \node[draw, circle] (v) [right of=w] {};

    \path[-]
        (y) edge [color=red] node[left] {} (x)
        (z) edge [color=blue] (x)
        (x) edge [color=blue] node[below] {} (w)
        (w) edge [color=red] node[right] {} (v)
    ;
\end{tikzpicture} \end{center}
\caption{The value of this position is also $\gameSet{0, -1}{-1, *} = \gameSet{0}{-1} = -\sfrac{1}{2} \pm \sfrac{1}{2}$.}
\label{fig:2pathFork}
\end{figure}

Another 4-edge path we will use later on is shown in figure \ref{fig:4path}.  Here Left's only option is to $-1$ and Right has moves to $-1$ and $*$.  

\begin{figure}[h!]
\begin{center}\begin{tikzpicture}[scale=1]
    \node[draw, circle] (y) {} ;
    \node[draw, circle] (x) [right of=y] {};
    \node[draw, circle] (w) [right of=x] {};
    \node[draw, circle] (v) [right of=w] {};
    \node[draw, circle] (u) [right of=v] {};

    \path[-]
        (y) edge [color=red] node[left] {} (x)
        (x) edge [color=blue] node[below] {} (w)
        (w) edge [color=red] node[right] {} (v)
        (v) edge [color=red] node[right] {} (u)
    ;
\end{tikzpicture} \end{center}
\caption{The value of this position is $\gameSet{-1}{-1} = -1*$.}
\label{fig:4path}
\end{figure}
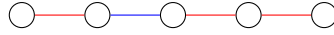

Adding a blue edge to give the path graph in figure \ref{fig:4path} an extra tail, as in figure \ref{fig:3pathFork}, doesn't change the value.

\begin{figure}[h!]
\begin{center}\begin{tikzpicture}[scale=1]
    \node[draw, circle] (y) {} ;
    \node[draw, circle] (z) [above of=y] {};
    \node[draw, circle] (x) [right of=y] {};
    \node[draw, circle] (w) [right of=x] {};
    \node[draw, circle] (v) [right of=w] {};
    \node[draw, circle] (u) [right of=v] {};

    \path[-]
        (z) edge [color=blue] (x)
        (y) edge [color=red] node[left] {} (x)
        (x) edge [color=blue] node[below] {} (w)
        (w) edge [color=red] node[right] {} (v)
        (v) edge [color=red] node[right] {} (u)
    ;
\end{tikzpicture} \end{center}
\caption{The value of this position is also $-1*$.}
\label{fig:3pathFork}
\end{figure}
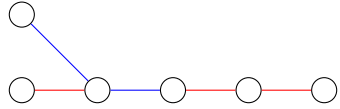

We will also use the 6-edge path shown in figure \ref{fig:6path}, which has value $-1$.  We've hit a point where it's helpful to show the derivation of these values.  For this one, the value is $\gameSet{-\sfrac{1}{2} \pm \sfrac{1}{2}}{-1*, -1*} = \gameSet{-\sfrac{1}{2} \pm \sfrac{1}{2}}{-1*} = -1$.

\begin{figure}[h!]
\begin{center}\begin{tikzpicture}[scale=1]
    \node[draw, circle] (y) {} ;
    \node[draw, circle] (x) [right of=y] {};
    \node[draw, circle] (w) [right of=x] {};
    \node[draw, circle] (v) [right of=w] {};
    \node[draw, circle] (u) [right of=v] {};
    \node[draw, circle] (t) [right of=u] {};
    \node[draw, circle] (s) [right of=t] {};

    \path[-]
        (y) edge [color=red] node[left] {} (x)
        (x) edge [color=blue] node[below] {} (w)
        (w) edge [color=red] node[right] {} (v)
        (v) edge [color=red] node[right] {} (u)
        (u) edge [color=blue] node[below] {} (t)
        (t) edge [color=red] node[right] {} (s)
    ;
\end{tikzpicture} \end{center}
\caption{The value of this position is $-1$.}
\label{fig:6path}
\end{figure}
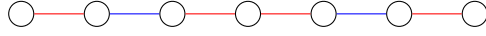

Adding a single blue edge to the graph in figure \ref{fig:6path} to give it another tail as in figure \ref{fig:5pathFork} also doesn't change the value there as well.  The derivation is $\gameSet{-1*, -\sfrac{1}{2} \pm \sfrac{1}{2}, -\sfrac{1}{2} \pm \sfrac{1}{2}}{-1*, 0,  -1*, -1*}$ $=$\\ $\gameSet{-1*, -\sfrac{1}{2} \pm \sfrac{1}{2}}{-1*} = -1$.

\begin{figure}[h!]
\begin{center}\begin{tikzpicture}[scale=1]
    \node[draw, circle] (y) {} ;
    \node[draw, circle] (z) [above of=y] {};
    \node[draw, circle] (x) [right of=y] {};
    \node[draw, circle] (w) [right of=x] {};
    \node[draw, circle] (v) [right of=w] {};
    \node[draw, circle] (u) [right of=v] {};
    \node[draw, circle] (t) [right of=u] {};
    \node[draw, circle] (s) [right of=t] {};

    \path[-]
        (y) edge [color=red] node[left] {} (x)
        (z) edge [color=blue] (x)
        (x) edge [color=blue] node[below] {} (w)
        (w) edge [color=red] node[right] {} (v)
        (v) edge [color=red] node[right] {} (u)
        (u) edge [color=blue] node[below] {} (t)
        (t) edge [color=red] node[right] {} (s)
    ;
\end{tikzpicture} \end{center}
\caption{The value of this position is also $-1$.}
\label{fig:5pathFork}
\end{figure}

Even if we add an edge to the \emph{other} end of  \ref{fig:5pathFork} to fork the tail on the opposite side, as in figure \ref{fig:4path2fork}, the value remains $-1$, because \\ $\gameSet{-1*, -\sfrac{1}{2} \pm \sfrac{1}{2}}{-1*, 0} = -1 $.

\begin{figure}[h!]
\begin{center}\begin{tikzpicture}[scale=1]
    \node[draw, circle] (y) {} ;
    \node[draw, circle] (z) [above of=y] {};
    \node[draw, circle] (x) [right of=y] {};
    \node[draw, circle] (w) [right of=x] {};
    \node[draw, circle] (v) [right of=w] {};
    \node[draw, circle] (u) [right of=v] {};
    \node[draw, circle] (t) [right of=u] {};
    \node[draw, circle] (s) [right of=t] {};
    \node[draw, circle] (s2) [above of=s] {};

    \path[-]
        (y) edge [color=red] node[left] {} (x)
        (z) edge [color=blue] (x)
        (x) edge [color=blue] node[below] {} (w)
        (w) edge [color=red] node[right] {} (v)
        (v) edge [color=red] node[right] {} (u)
        (u) edge [color=blue] node[below] {} (t)
        (t) edge [color=red] node[right] {} (s)
        (t) edge [color=blue] (s2)
    ;
\end{tikzpicture} \end{center}
\caption{The value of this position is also $-1$.}
\label{fig:4path2fork}
\end{figure}

We will also use a reverse version of $-A_2$ with additional edges, shown in figure \ref{fig:a2plus}, which also has a value of $0$.  The labeling of these gadgets in the figures is added to match up with the clause gadgets in figure \ref{fig:clause} in the next section.

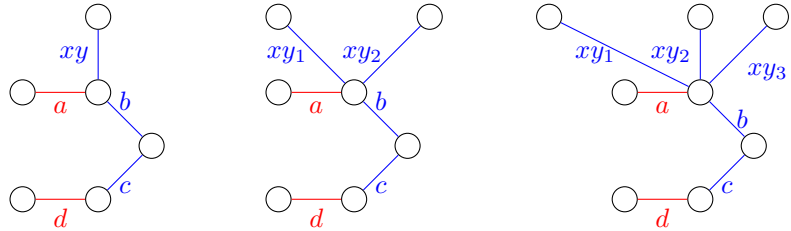
\begin{figure}[h!]
\begin{center}\begin{tikzpicture}[scale=1]
    \node[draw, circle] (y) {} ;
    \node[draw, circle] (x1) at (0, 1) {};
    \node[draw, circle] (aEnd) [left of=y] {};
    \node[draw, circle] (bc) [below right of=y] {};
    \node[draw, circle] (cd) [below left of=bc] {};
    \node[draw, circle] (dEnd) [left of=cd] {};

    \path[-]
        (y) edge [color=blue] node[left] {$xy$} (x1)
        (y) edge [color=red] node[below] {$a$} (aEnd)
        (y) edge [color=blue] node[above] {$b$} (bc)
        (bc) edge [color=blue] node[below] {$c$} (cd)
        (cd) edge [color=red] node[below] {$d$} (dEnd)
    ;
\end{tikzpicture} 
\hspace{1cm}
\begin{tikzpicture}[scale=1]
    \node[draw, circle] (y) {} ;
    \node[draw, circle] (x1) at (-1, 1) {};
    \node[draw, circle] (x2) at (1, 1) {};
    \node[draw, circle] (aEnd) [left of=y] {};
    \node[draw, circle] (bc) [below right of=y] {};
    \node[draw, circle] (cd) [below left of=bc] {};
    \node[draw, circle] (dEnd) [left of=cd] {};

    \path[-]
        (y) edge [color=blue] node[left] {$xy_1$} (x1)
        (y) edge [color=blue] node[left] {$xy_2$} (x2)
        (y) edge [color=red] node[below] {$a$} (aEnd)
        (y) edge [color=blue] node[above] {$b$} (bc)
        (bc) edge [color=blue] node[below] {$c$} (cd)
        (cd) edge [color=red] node[below] {$d$} (dEnd)
    ;
\end{tikzpicture} 
\hspace{1cm}
\begin{tikzpicture}[scale=1]
    \node[draw, circle] (y) {} ;
    \node[draw, circle] (x1) at (-2, 1) {};
    \node[draw, circle] (x2) at (0, 1) {};
    \node[draw, circle] (x3) at (1, 1) {};
    \node[draw, circle] (aEnd) [left of=y] {};
    \node[draw, circle] (bc) [below right of=y] {};
    \node[draw, circle] (cd) [below left of=bc] {};
    \node[draw, circle] (dEnd) [left of=cd] {};

    \path[-]
        (y) edge [color=blue] node[left] {$xy_1$} (x1)
        (y) edge [color=blue] node[left] {$xy_2$} (x2)
        (y) edge [color=blue] node[below right] {$xy_3$} (x3)
        (y) edge [color=red] node[below] {$a$} (aEnd)
        (y) edge [color=blue] node[right] {$b$} (bc)
        (bc) edge [color=blue] node[below] {$c$} (cd)
        (cd) edge [color=red] node[below] {$d$} (dEnd)
    ;
\end{tikzpicture} \end{center}
\caption{The value of each of these positions is $0$.  Adding more single edges $xy$ does not change the value, because playing on one of those edges is equivalent to playing on all of them.  The form for all of them is $\gameSet{*, -1, *}{*, *} = 0$.}
\label{fig:a2plus}
\end{figure}

Combining $b$ and $c$ as one edge gives the graph shown in figure \ref{fig:FPiece}

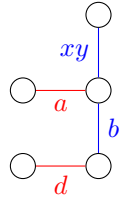
\begin{figure}[h!]
\begin{center}\begin{tikzpicture}[scale=1]
    \node[draw, circle] (y) {} ;
    \node[draw, circle] (x1) at (0, 1) {};
    \node[draw, circle] (aEnd) [left of=y] {};
    \node[draw, circle] (bc) [below of=y] {};
    \node[draw, circle] (dEnd) [left of=bc] {};

    \path[-]
        (y) edge [color=blue] node[left] {$xy$} (x1)
        (y) edge [color=red] node[below] {$a$} (aEnd)
        (y) edge [color=blue] node[right] {$b$} (bc)
        (bc) edge [color=red] node[below] {$d$} (dEnd)
    ;
\end{tikzpicture}  \end{center}
\caption{The value of this position is $\gameSet{0, -1}{*, -1} = \gameSet{0}{-1} = -\sfrac{1}{2} \pm \sfrac{1}{2}$.}
\label{fig:FPiece}
\end{figure}

Instead of merging two edges, we can modify $-A_2$ another way by adding a red edge to the middle, as in figure \ref{fig:a2plusMid}, which has a value of $-\sfrac{3}{2} \pm \sfrac{1}{2}$.

\begin{figure}[h!]
\begin{center}\begin{tikzpicture}[scale=1]
    \node[draw, circle] (y) {} ;
    \node[draw, circle] (aEnd) [left of=y] {};
    \node[draw, circle] (bc) [below right of=y] {};
    \node[draw, circle] (cd) [below left of=bc] {};
    \node[draw, circle] (dEnd) [left of=cd] {};
    \node[draw, circle] (s) [right of=bc] {};

    \path[-]
        (y) edge [color=red] node[below] {$a$} (aEnd)
        (y) edge [color=blue] node[above] {$b$} (bc)
        (bc) edge [color=blue] node[below] {$c$} (cd)
        (cd) edge [color=red] node[below] {$d$} (dEnd)
        (bc) edge [color=red] node[below] {} (s)
    ;
\end{tikzpicture} \end{center}
\caption{The value of this position is $\gameSet{-1, -1}{-\sfrac{1}{2}\pm\sfrac{1}{2}, -2,-\sfrac{1}{2}\pm\sfrac{1}{2}} = \gameSet{-1}{-2} = -\sfrac{3}{2} \pm \sfrac{1}{2}$.}
\label{fig:a2plusMid}
\end{figure}
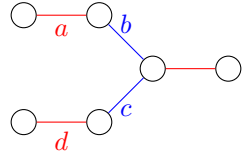

If we add a path of two edges to $-A_2$ instead, we get the graph in \ref{fig:a2plus2}, which has a value of $-2*$.

\begin{figure}[h!]
\begin{center}\begin{tikzpicture}[scale=1]
    \node[draw, circle] (y) {} ;
    \node[draw, circle] (aEnd) [left of=y] {};
    \node[draw, circle] (bc) [below right of=y] {};
    \node[draw, circle] (cd) [below left of=bc] {};
    \node[draw, circle] (dEnd) [left of=cd] {};
    \node[draw, circle] (s) [right of=bc] {};
    \node[draw, circle] (s2) [right of=s] {};

    \path[-]
        (y) edge [color=red] node[below] {$a$} (aEnd)
        (y) edge [color=blue] node[above] {$b$} (bc)
        (bc) edge [color=blue] node[below] {$c$} (cd)
        (cd) edge [color=red] node[below] {$d$} (dEnd)
        (bc) edge [color=red] node[below] {} (s)
        (s) edge [color=red] (s2)
    ;
\end{tikzpicture} \end{center}
\caption{The value of this position is $\gameSet{-2, -2}{1*, 1+ *, -2, 0} = \gameSet{-2}{-2} = -2*$.}
\label{fig:a2plus2}
\end{figure}
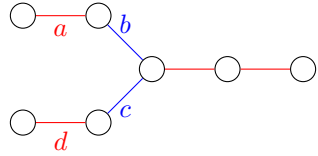

In figure \ref{fig:oneClauseS}, we show and find the value of a graph that adds an edge from figure \ref{fig:a2plus}.

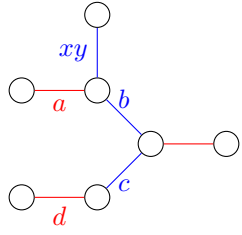
\begin{figure}[h!]
\begin{center} \begin{tikzpicture}[scale=1]
    \node[draw, circle] (y) {} ;
    \node[draw, circle] (x1) at (0, 1) {};
    \node[draw, circle] (aEnd) [left of=y] {};
    \node[draw, circle] (bc) [below right of=y] {};
    \node[draw, circle] (cd) [below left of=bc] {};
    \node[draw, circle] (dEnd) [left of=cd] {};
    \node[draw, circle] (s) [right of=bc] {};

    \path[-]
        (y) edge [color=blue] node[left] {$xy$} (x1)
        (y) edge [color=red] node[below] {$a$} (aEnd)
        (y) edge [color=blue] node[above] {$b$} (bc)
        (bc) edge [color=blue] node[below] {$c$} (cd)
        (cd) edge [color=red] node[below] {$d$} (dEnd)
        (bc) edge [color=red] node[below] {} (s)
    ;
\end{tikzpicture}  \end{center}
\caption{The value of this is $\gameSet{-\sfrac{1}{2} \pm \sfrac{1}{2}, *}{-1*, -\sfrac{1}{2} \pm \sfrac{1}{2}} = \gameSet{*}{-1*} = -\sfrac{1}{2} \pm \sfrac{1}{2}*$.}
\label{fig:oneClauseS}
\end{figure}

\begin{figure}[h!]
\begin{center} \begin{tikzpicture}[scale=1]
    \node[draw, circle, minimum width=2cm] (G) {$G$} ;
    \node[draw, fill=white, circle] (g1) at (.707, .707) {};
    \node[draw, fill=white, circle] (g2) at (1, 0) {};
    \node[draw, fill=white, circle] (g3) at (.707, -.707) {};
    \node[draw, fill=white, circle] (r1) at (2,0) {};
    \node[draw, fill=white, circle] (r2) at (3,0) {};

    \path[-]
        (g1) edge [color=blue]  (r1)
        (g2) edge [color=blue]  (r1)
        (g3) edge [color=blue]  (r1)
        (r1) edge [color=red] (r2)
    ;
\end{tikzpicture}  \end{center}
\caption{The value of this is always greater than or equal $G - 1$, by lemma \ref{lem:deathStar}, no matter how many blue edges connect $G$ to to the red edge.}
\label{fig:deathStar}
\end{figure}
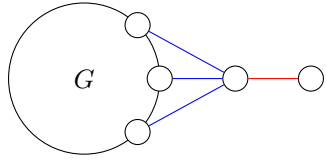

\begin{lemma}[Death Star Lemma]
    Consider modifying a \parck{} position, $G$, to get $G'$ in the following way: we add a new vertex and connect it to $G$ with zero or more blue edges and also add a new pendant red edge off that vertex (as in figure \ref{fig:deathStar}), then $G' \geq G - 1$.
    \label{lem:deathStar}
\end{lemma}

\begin{proof}
    If no blue edges were added, then we're done because $G' = G - 1$.  If there are blue edges, then we have some further work to do.
    
    In that case, we prove, equivalently, that $G' - G + 1 \geq 0$ by inductively showing that Right loses going first.  Right has three different places they may be able to play:
    \begin{itemize}
        \item[] $-G$: Left can respond with the associated play in $G'$.  Now Right loses going first by the inductive hypothesis, even if Left's play removed all the added blue edges.

        \item[] Added red edge in $G'$: the game is now $G - G + 1 = 1$, so Left wins.
        \item[] Another red edge in $G'$: then Left makes the associated move in $-G$ and Right loses going first by the inductive hypothesis, even if Right's play removed all the added blue edges.
    \end{itemize}
\end{proof}

\section{\cclass{PSPACE} Reduction}

We reduce directly from \ruleset{Positive CNF}.  There are two different gadgets used to complete this reduction: one for variables and one for clauses. Throughout we assume there are an odd number of variables. If not, we add an extra variable not included in any clause. Similarly, we assume there is more than one clause. If not, we simply clone the solitary clause in the final construction with no effect on the outcome of the cnf itself.

The variable gadget is described in figure \ref{fig:variable}.  These gadgets are hot because playing on the T and F edges changes the value of the $-A_2$ portion, incentivizing play on those edges.  If Left plays on a T, that removes all three of the connected red edges as well, leaving two unconnected blue edges and the individual unconnected red edge for a total value of $1$.  If Right plays on an F, then all that remains is the entire $-A_2$ part and the separate red edge, for a total value of $-1$.  Thus these gadgets each act as a $\pm 1$ switch. 

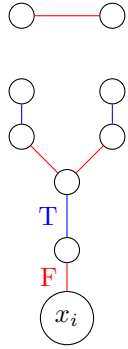
\begin{figure}[h!]
\begin{center}\begin{tikzpicture}[scale=.6]

    \node[draw, circle] (nA2) at (0,0) {} ;
    \node[draw, circle] (A21) at (-1,1) {};
    \node[draw, circle] (A22) at (1,1) {};
    \node[draw, circle] (A23) at (-1,2) {};
    \node[draw, circle] (A24) at (1,2) {};
    \node[draw, circle] (varmid) at (0, -1.5) {};
    \node[draw, circle] (varx) at (0, -3.0) {$x_i$};
    \node[draw, circle] (extra1) [above of=A23] {};
    \node[draw, circle] (extra2) [above of=A24] {};

    \path[-]
        (varx) edge [color=red] node[left] {F} (varmid)
        (varmid) edge [color=blue] node[left] {T} (nA2)
        (extra1) edge [color=red] (extra2)
        (nA2) edge [color=red] (A21)
        (nA2) edge [color=red] (A22)
        (A21) edge [color=blue] (A23)
        (A22) edge [color=blue] (A24)
    ;
\end{tikzpicture} \end{center}
\caption{Variable Gadget.  True plays on the blue edge marked T to choose variable $x_i$, whereas false plays on the red edge F.}
\label{fig:variable}
\end{figure}

The clause gadgets, as described in \ref{fig:clause}, are each connected to the $x_i$ vertices of the variables in them, as well as to a single master $s$ vertex (see figure \ref{fig:fullGraph}).  If clause $j$ loses all connections to the $x_i$ vertices and $s_j$ is chosen by Right, also removing $b_j$ and $c_j$, then only the disconnected edges $a_j$ and $d_j$ remain, for a value of -2.

\begin{figure}[h!]
\begin{center} 
\begin{tikzpicture}[scale=.8]    
    \node[draw, circle] (x1) at (0,10) {$x_1$};
    \node[draw, circle] (x2) at (4,10) {$x_2$};
    \node[draw, circle] (x3) at (8,10) {$x_3$};
    \node[draw, circle] (x4) at (12,10) {$x_4$};
    \node[draw, circle] (x5) at (16,10) {$x_5$};

    \node[draw, circle] (r11) at (-1,15) {};
    \node[draw, circle] (r12) at (1,15) {};
    \node[draw, circle] (r21) at (3,15) {};
    \node[draw, circle] (r22) at (5,15) {};
    \node[draw, circle] (r31) at (7,15) {};
    \node[draw, circle] (r32) at (9,15) {};
    \node[draw, circle] (r41) at (11,15) {};
    \node[draw, circle] (r42) at (13,15) {};
    \node[draw, circle] (r51) at (15,15) {};
    \node[draw, circle] (r52) at (17,15) {};    
    
    \node[draw, circle] (r1) at (0,11) {};
    \node[draw, circle] (s1) at (0,12) {};    
    \node[draw, circle] (t1) at (-1,13) {};    
    \node[draw, circle] (u1) at (1,13) {};    
    \node[draw, circle] (v1) at (-1,14) {};
    \node[draw, circle] (w1) at (1,14) {};    

    \node[draw, circle] (r2) at (4,11) {};
    \node[draw, circle] (s2) at (4,12) {};    
    \node[draw, circle] (t2) at (3,13) {};    
    \node[draw, circle] (u2) at (5,13) {};    
    \node[draw, circle] (v2) at (3,14) {};
    \node[draw, circle] (w2) at (5,14) {};

    \node[draw, circle] (r3) at (8,11) {};
    \node[draw, circle] (s3) at (8,12) {};    
    \node[draw, circle] (t3) at (7,13) {};    
    \node[draw, circle] (u3) at (9,13) {};    
    \node[draw, circle] (v3) at (7,14) {};
    \node[draw, circle] (w3) at (9,14) {};

    \node[draw, circle] (r4) at (12,11) {};
    \node[draw, circle] (s4) at (12,12) {};    
    \node[draw, circle] (t4) at (11,13) {};    
    \node[draw, circle] (u4) at (13,13) {};    
    \node[draw, circle] (v4) at (11,14) {};
    \node[draw, circle] (w4) at (13,14) {}; 

    \node[draw, circle] (r5) at (16,11) {};
    \node[draw, circle] (s5) at (16,12) {};    
    \node[draw, circle] (t5) at (15,13) {};    
    \node[draw, circle] (u5) at (17,13) {};    
    \node[draw, circle] (v5) at (15,14) {};
    \node[draw, circle] (w5) at (17,14) {};     
    
    \node[draw, circle] (c1) at (3,3) {};
    \node[draw, circle] (c2) at (9,3) {};    
    \node[draw, circle] (c3) at (15,3) {};
    
    \node[draw, circle, above left=of c1] (a1) {};
    \node[draw, circle, left=of a1] (b1) {};
    \node[draw, circle, left=of c1] (d1) {};
    \node[draw, circle, left=of d1] (e1) {};

    \node[draw, circle, above left=of c2] (a2) {};
    \node[draw, circle, left=of a2] (b2) {};
    \node[draw, circle, left=of c2] (d2) {};
    \node[draw, circle, left=of d2] (e2) {};
    
    \node[draw, circle, above left=of c3] (a3) {};
    \node[draw, circle, left=of a3] (b3) {};
    \node[draw, circle, left=of c3] (d3) {};
    \node[draw, circle, left=of d3] (e3) {};    
    \node[draw, circle] (s) at (6,0) {$s$};
    
    \path[-]
    	(c1) edge [color=blue] (a1)
	(c1) edge [color=blue] (d1)
	(a1) edge [color=red] (b1)
	(d1) edge [color=red] (e1)
	
    	(c2) edge [color=blue] (a2)
	(c2) edge [color=blue] (d2)
	(a2) edge [color=red] (b2)
	(d2) edge [color=red] (e2)
	
    	(c3) edge [color=blue] (a3)
	(c3) edge [color=blue] (d3)
	(a3) edge [color=red] (b3)
	(d3) edge [color=red] (e3)	
	
	(s) edge [color=red] (c1)
	(s) edge [color=red] (c2)
	(s) edge [color=red] (c3)	
	
	(x1) edge [color=blue] (a1)
	(x2) edge [color=blue] (a1)	

	(x2) edge [color=blue] (a2)
	(x3) edge [color=blue] (a2)				
	(x4) edge [color=blue] (a2)					

	(x1) edge [color=blue] (a3)
	(x3) edge [color=blue] (a3)
	
	(x1) edge [color=red] (r1)
	(r1) edge [color=blue] (s1)
	(s1) edge [color=red] (t1)
	(s1) edge [color=red] (u1)	
	(t1) edge [color=blue] (v1)
	(u1) edge [color=blue] (w1)
	
	(x2) edge [color=red] (r2)
	(r2) edge [color=blue] (s2)
	(s2) edge [color=red] (t2)
	(s2) edge [color=red] (u2)	
	(t2) edge [color=blue] (v2)
	(u2) edge [color=blue] (w2)
	
	(x3) edge [color=red] (r3)
	(r3) edge [color=blue] (s3)
	(s3) edge [color=red] (t3)
	(s3) edge [color=red] (u3)	
	(t3) edge [color=blue] (v3)
	(u3) edge [color=blue] (w3)
	
	(x4) edge [color=red] (r4)
	(r4) edge [color=blue] (s4)
	(s4) edge [color=red] (t4)
	(s4) edge [color=red] (u4)	
	(t4) edge [color=blue] (v4)
	(u4) edge [color=blue] (w4)		

	(x5) edge [color=red] (r5)
	(r5) edge [color=blue] (s5)
	(s5) edge [color=red] (t5)
	(s5) edge [color=red] (u5)	
	(t5) edge [color=blue] (v5)
	(u5) edge [color=blue] (w5)		    

    (r11) edge [color=red] (r12)
    (r21) edge [color=red] (r22)
    (r31) edge [color=red] (r32)
    (r41) edge [color=red] (r42)
    (r51) edge [color=red] (r52)
    ;
\end{tikzpicture} 
\end{center}
\caption{The construction for the positive cnf $(x_1 \vee x_2) \wedge (x_2 \vee x_3 \vee x_4) \wedge (x_1 \vee x_3)$. An extra variable gadget is added to ensure an odd number.}
\label{fig:fullGraph}
\end{figure}
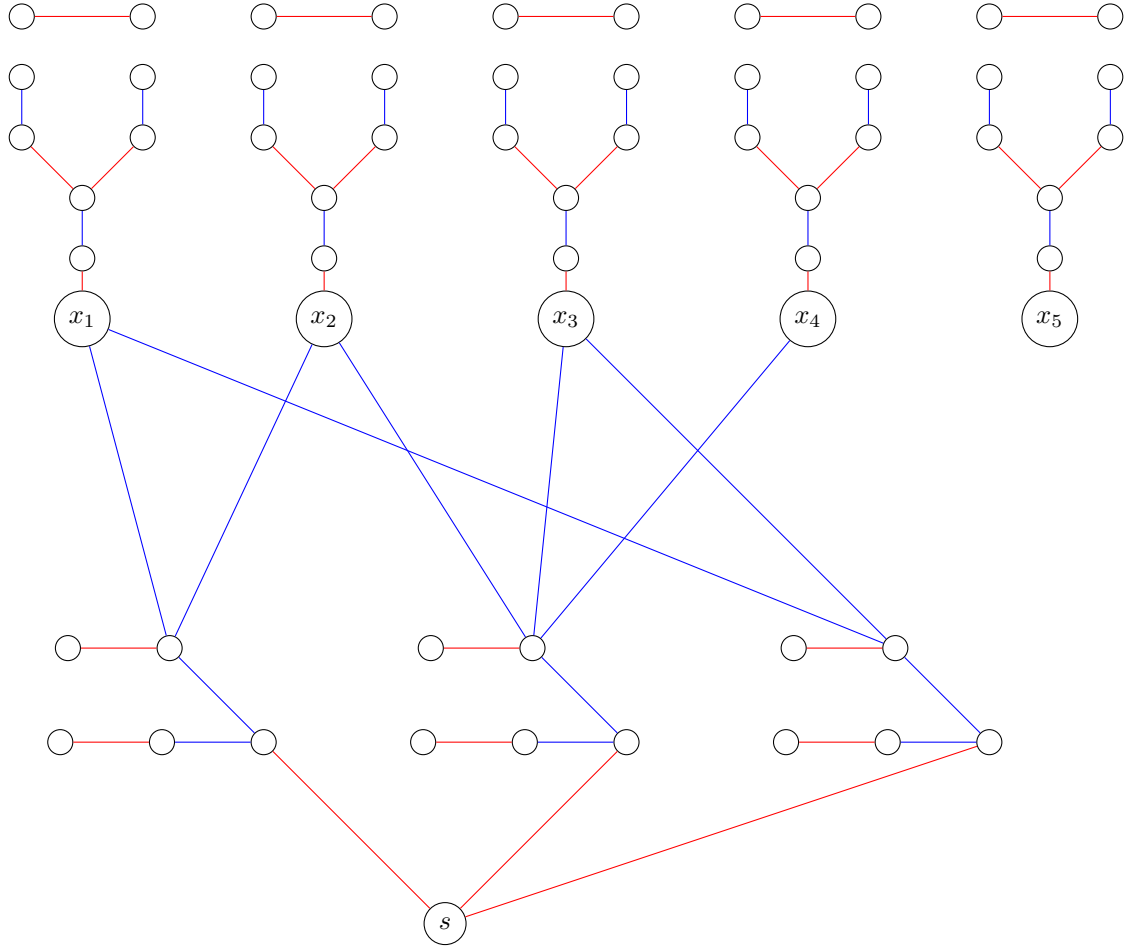

Playing on $s_j$ will also remove all other red edges $s_k$ ($k \neq j$).  This leaves a position with value zero, as shown in figure \ref{fig:a2plus}.

\begin{figure}[h!]
\begin{center}\begin{tikzpicture}[scale=1]
    \node[draw, circle] (y) {$y_j$} ;
    \node[draw, circle] (x1) at (-3, 3.5) {$x_f$};
    \node[draw, circle] (x2) at (-1, 3.5) {$x_g$};
    \node[draw, circle] (x3) at (1.5, 3.5) {$x_h$};
    \node[draw, circle] (aEnd) at (-1.5, 0) {};
    \node[draw, circle] (bc) at (1,-1) {};
    \node[draw, circle] (cd) [left of=bc] {};
    \node[draw, circle] (dEnd) at (-1.5, -1) {};
    \node[draw, circle] (s) at (2, -2) {$s$};

    \path[-]
        (y) edge [color=blue] node[left] {$x_fy_j$} (x1)
        (y) edge [color=blue] node[left] {$x_gy_j$} (x2)
        (y) edge [color=blue] node[left] {$x_hy_j$} (x3)
        (y) edge [color=red] node[below] {$a_j$} (aEnd)
        (y) edge [color=blue] node[above] {$b_j$} (bc)
        (bc) edge [color=blue] node[below] {$c_j$} (cd)
        (cd) edge [color=red] node[below] {$d_j$} (dEnd)
        (bc) edge [color=red] node[below] {$s_j$} (s)
    ;
\end{tikzpicture} \end{center}
\caption{Clause Gadget, consisting of edges $a_j$, $b_j$, $c_j$, $d_j$, and $s_j$.  For each clause,$s_j$ which connects the clause gadget to the singular node $s$.  For each variable in clause $j$, $x_i$, there is an edge, $x_iy_j$ connecting node $x_i$ and $y_j$.}
\label{fig:clause}
\end{figure}
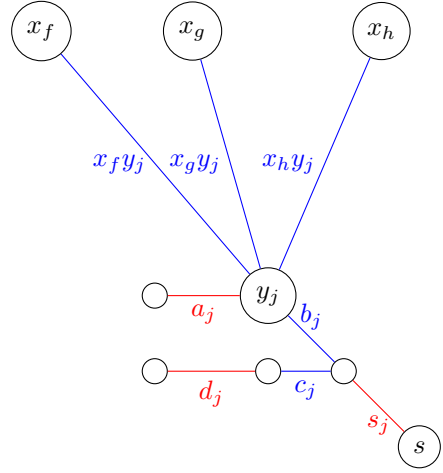

The value of a single clause gadget alone, including the red edge to $s$, is $-\sfrac{1}{2} + \sfrac{1}{2}*$ as long as it's connected to at least one variable vertex.

We need to assure that players restrict their moves to the variable gadgets until they are all chosen True or False.

\begin{lemma}
    Once all variables have been claimed, if $s$ and all edges $s_j$ are gone from all clause gadgets, then the graph has value $0$, whether or not any of the clauses are still connected to variable nodes.
    \label{lem:noS}
\end{lemma}

\begin{proof}
Proof by induction on the number of clauses.  The base case of a single clause is either the reverse of $-A_2$, which has value $0$, or there are one or more blue edges as in figure \ref{fig:a2plus}, which also has value $0$.

We prove the inductive case by showing that the first player to move loses. Without loss of generality we assume play occurs on clause $j$.  First, we cover Left's options:
\begin{itemize}
    \item If Left chooses an $x_iy_j$, that also removes edges $a_j$ and $b_j$ as well as any other edges $x_iy_l$ and $x_ky_j$.  The remaining $c_j$, $d_j$ pair has value $*$.  By the inductive hypothesis, the rest of the graph has value $0$.  $0 + * = *$ is a win for the next player.
    \item If Left chooses $b_j$, then only $d_j$ remains from this clause.  By the inductive hypothesis, the total remaining graph has value $-1$, a loss for Left.
    \item If Left chooses $c_j$, then Right can choose $a_j$ to remove the rest of the clause.  The remaining graph has value zero by the inductive hypothesis. 
\end{itemize}
For Right's options:
\begin{itemize}
    \item If Right chooses $a_j$, then just like when Left started by choosing $x_iy_j$, the value of the graph is now $*$, so Left can win.
    \item If Right chooses $d_j$, Left can respond with $b_j$ to win by the inductive hypothesis.
\end{itemize}
All moves are losing moves, so the game has value $0$.
\end{proof}

\begin{figure}[h!]
\begin{center} 
\begin{tikzpicture}[scale=.6]    
    \node[draw, circle] (x1) at (0,10) {};
    \node[draw, circle] (x2) at (4,10) {};
    \node[draw, circle] (x3) at (8,10) {};
    \node[draw, circle] (x4) at (12,10) {};    
    
    \node[draw, circle] (c1) at (3,3) {};
    \node[draw, circle] (c2) at (9,3) {};    
    \node[draw, circle] (c3) at (15,3) {};
    
    \node[draw, circle, above left=of c1] (a1) {};
    \node[draw, circle, left=of a1] (b1) {};
    \node[draw, circle, left=of c1] (d1) {};
    \node[draw, circle, left=of d1] (e1) {};

    \node[draw, circle, above left=of c2] (a2) {};
    \node[draw, circle, left=of a2] (b2) {};
    \node[draw, circle, left=of c2] (d2) {};
    \node[draw, circle, left=of d2] (e2) {};
    
    \node[draw, circle, above left=of c3] (a3) {};
    \node[draw, circle, left=of a3] (b3) {};
    \node[draw, circle, left=of c3] (d3) {};
    \node[draw, circle, left=of d3] (e3) {};    
    \node[draw, circle] (s) at (6,0) {$s$};
    
    \path[-]
    (c1) edge [color=blue] (a1)
	(c1) edge [color=blue] (d1)
	(a1) edge [color=red] (b1)
	(d1) edge [color=red] (e1)
	
    (c2) edge [color=blue] (a2)
	(c2) edge [color=blue] (d2)
	(a2) edge [color=red] (b2)
	(d2) edge [color=red] (e2)
	
    (c3) edge [color=blue] (a3)
	(c3) edge [color=blue] (d3)
	(a3) edge [color=red] (b3)
	(d3) edge [color=red] (e3)	
	
	(s) edge [color=red] (c1)
	(s) edge [color=red] (c2)
	(s) edge [color=red] (c3)	
	
	(x1) edge [color=blue] (a1)
	(x2) edge [color=blue] (a1)	

	(x2) edge [color=blue] (a2)
	(x3) edge [color=blue] (a2)				
	(x4) edge [color=blue] (a2)					

	(x1) edge [color=blue] (a3)
	(x3) edge [color=blue] (a3)	    

    ;
\end{tikzpicture} 
\end{center}
\caption{An example of a position remaining after all variables have been claimed. Not shown are the isolated red edges from the variable gadgets. The value is $\{*|-1*\}$ when one clause is present, and $-1$ if there are more.}
\label{fig:noMoreVars}
\end{figure}
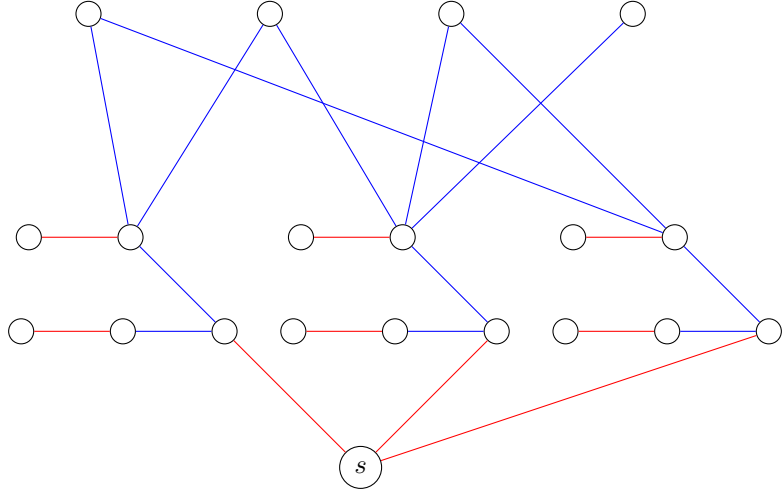

\begin{corollary}
    If Left chooses $c_j$ from a position without $s$, as in lemma \ref{lem:noS}, the value of the game is $*$. 
    \label{cor:noSStar}
\end{corollary}

\begin{proof}
$c_j$ was chosen, so only $b_j$ and $d_j$ are also missing from what was left of clause $j$.

Proof by induction.  The base case of one clause is just a single $a_j$ connected to $x_iy_j$ edges, which is $*$.  

In the recursive case, each player's best move is to $0$.  If Left chooses one of the $x_iy_j$ edges, then the position is zero by lemma \ref{lem:noS}.  If they choose anything else, Right can always at least return to $*$ by induction by using the same moves spelled out in lemma \ref{lem:noS}, unless they can isolate $a_j$ in which case the position has value -1. 

If Right chooses $a_j$, the position is zero by lemma \ref{lem:noS}.  If they play anywhere else, Left can return to $*$ by induction by following the responses from lemma \ref{lem:noS}.

The best moves for each player are to $0$, so the value is $*$.
\end{proof}

\begin{theorem}
    For any \parck{} position consisting of clause gadgets connected to individual variable $x_i$ vertices, if each $y_j$ is connected to at least one $x_i$, then the value of the whole graph is -1 unless there is only one clause gadget, in which case the value is $-\sfrac{1}{2} \pm \sfrac{1}{2}*$.
    \label{thm:allPositive}
\end{theorem}

\begin{proof}
Proof by induction.  We require as base case both one clause and two clause positions. The one clause case is covered by earlier analysis that the value of a single clause is $-\sfrac{1}{2} \pm \sfrac{1}{2}*$ as shown in figure \ref{fig:oneClauseS}. For two clauses, an exhaustive analysis yields $-1$ irrespective of the number of variable nodes. 


In our recursive case, we assume that we are working with a graph with three or more clauses, each of which is connected to at least one variable.  See figure \ref{fig:clause} for reference. Without loss of generality, we analyze each player's move options on some clause $j$ and show that that player loses on the graph with a single blue edge added.  This will show that the sum of the graph plus the blue edge is zero, which proves our claim.  First, for Left:

\begin{itemize}
    \item Left chooses the single blue edge we added.  Right can win by taking $s_j$.  This leaves an isolated $d_j$ with value -1 and the rest of the graph has value $*$ by corollary \ref{cor:noSStar}.  $-1*$ is a winning value for Right.
    \item Left chooses edge $x_iy_j$, removing $a_j$, $b_j$, and whatever other $xy$ edges were connected to $x_i$ or $y_j$.  Right can win by responding with $s_j$.  This leaves the isolated red edge $d_j$, which cancels out the disconnected blue edge.  Since $s$ has been deleted, the remainder of the graph is a case of lemma \ref{lem:noS}, which has value zero.  
    \item Left chooses edge $b_j$, also removing $a_j$, $c_j$, $s_j$, and all edges of the form $x_iy_j$, leaving only $d_j$ from the clause.  By the inductive hypothesis, the remainder of the clause-gadgets-graph has value -1.  Adding in the remaining blue edge and the $d_j$ gives us a value of -1 for the entire graph, so Left's move causes them to lose.
    \item Left chooses $c_j$, also removing $b_j$, $d_j$, and $s_j$.  Right can respond by taking $a_j$, which removes all edges $x_iy_j$.  By the inductive hypothesis the remaining graph has value $0$ when including the extra blue edge, so Right has a winning response.
\end{itemize}

We continue by considering the options for Right:

\begin{itemize}
    \item Right chooses edge $a_j$, also removing $b_j$ and all edges $x_iy_j$.  Left can respond by choosing $c_j$, which just removes all other edges from this clause gadget.  By the inductive hypothesis the remaining graph has value zero with the extra blue edge, so Left wins.
    \item Right chooses $d_j$, also removing $c_j$.  Left can win by choosing $b_j$, which removes all remaining edges from the clause.  By the inductive hypothesis and with the blue edge, the graph has value zero.
    \item Right chooses $s_j$, also removing $b_j$, $c_j$ and all other $s_l$ edges.  By corollary \ref{cor:noSStar}, the remainder of the graph has value $*$, since $d_j$ cancels out the extra blue edge.  Left can win on $*$.  
\end{itemize}

All options for our graph plus $1$ are losing options, so the value of the graph in the inductive case must be -1.

\end{proof}

\begin{lemma}
The value of a \parck\ position with all variables claimed, $s$ present, and at least one false clause is strictly less than $-1$. \label{lem:lessThanNegOne}
\end{lemma}
\begin{proof}
Say the cnf is false and there is at least one clause gadget, $i$, unconnected to a variable node. Let $a_j, b_j, c_j, d_j$ be edges in clause $j$ as in figure \ref{fig:clause}, let $x_ky_j$ be a blue edge from variable node $k$ to clause $i$, and let $s_i$ be the red edge from $s$ to clause $i$. We consider the position in the statement of the claim added to $1$ (an isolated blue edge) and show this sum is negative. First consider Right's responses when Left moves first.

\begin{itemize}
    \item[$x_ky_j$] Right responds with $s_i$ to leave $-1*$.
    \item[$b_j$] Right responds with $s_i$ to leave $-2$.
    \item[$b_i$] Right responds with $s_j, i\neq j$ to leave $-2+1$ along with a red edge attached to a position of value $0$, a sum of at most $-1$. Since we assume there are always at least two clauses, $s_j$ exists.
    \item[$c_j$] Right responds with $s_i$, again leaving $-2+1 = -1$ along with a position of value $0$. The same argument holds if Left chooses $c_i$.    
\end{itemize}
If Left plays the isolated blue edge, then Right responds with $s_i$ to leave $-2$. So Left loses moving first.

If Right moves first, they can take $s_i$ to a sum of $-1$.

Since Right wins no matter who goes first, the value of the position itself is strictly less than $-1$.
\end{proof}

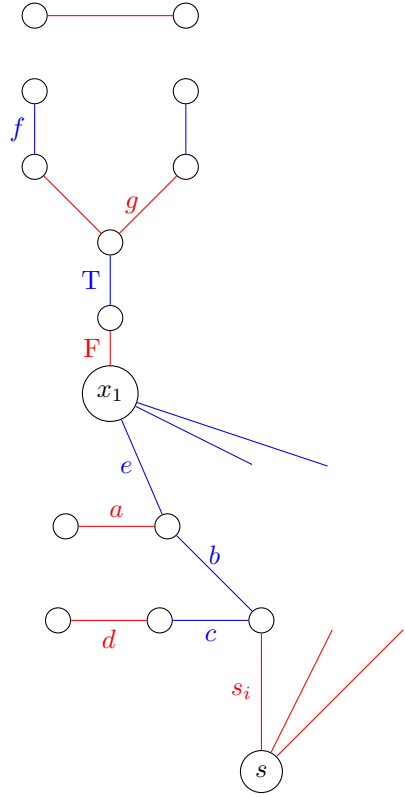
\begin{figure}[h!]
\begin{center} 
\begin{tikzpicture}[scale=1]    
    \node[draw, circle] (x1) at (0,6) {$x_1$};
    
    \node[draw, circle] (r1) at (0,7) {};
    \node[draw, circle] (s1) at (0,8) {};    
    \node[draw, circle] (t1) at (-1,9) {};    
    \node[draw, circle] (u1) at (1,9) {};    
    \node[draw, circle] (v1) at (-1,10) {};
    \node[draw, circle] (w1) at (1,10) {}; 
    \node (a2) at (2,5) {};
    \node (a3) at (3,5) {};
    \node (c2) at (3,3) {};
    \node (c3) at (4,3) {};
    \node[draw, circle] (red1) at (-1,11) {};
    \node[draw, circle] (red2) at (1,11) {};
    
    \node[draw, circle] (c1) at (2,3) {};
    
    \node[draw, circle, above left=of c1] (a1) {};
    \node[draw, circle, left=of a1] (b1) {};
    \node[draw, circle, left=of c1] (d1) {};
    \node[draw, circle, left=of d1] (e1) {};

    \node[draw, circle] (s) at (2,1) {$s$};

    \path[-]
    (c1) edge [color=blue] node[above] {$b$} (a1)
	(c1) edge [color=blue] node[below] {$c$} (d1)
	(a1) edge [color=red] node[above] {$a$} (b1)
	(d1) edge [color=red] node[below] {$d$} (e1)
	
	(s) edge [color=red] node[left] {$s_i$} (c1)
	
	(x1) edge [color=blue] node[left] {$e$} (a1)			
	
	(x1) edge [color=red] node[left] {F} (r1) 
	(r1) edge [color=blue] node[left] {T} (s1)
	(s1) edge [color=red] (t1)
	(s1) edge [color=red] node[left] {$g$} (u1)	
	(t1) edge [color=blue] node[left] {$f$} (v1)
	(u1) edge [color=blue] (w1)

    (x1) edge [color=blue] (a2)
    (x1) edge [color=blue] (a3)
    (s) edge [color=red] (c2)
    (s) edge [color=red] (c3)

    (red1) edge [color=red] (red2)
			
    ;
\end{tikzpicture} 
\end{center}
\caption{Part of a reduced cnf position showing one clause gadget and one variable gadget}
\label{fig:var_temp}
\end{figure}

We will show that if Right is able to disconnect all variables from one clause, then those two points is enough to cause them to win.  In order to do that, we first need to show what happens if none of them are disconnected. First, we will examine positions without $s$.

\begin{lemma}
Consider a graph that is the reduction of a positive cnf. If any variable gadgets remain then the best move for both players is to choose a T or F edge, as appropriate. \label{lem:vars_first}
\end{lemma}
\begin{proof}
We proceed by the method of disproving the existence of a minimal counter-example, first on the number of unclaimed variables, and then by the number of clauses remaining. For the base case, consider the graph with one variable gadget and one clause gadget in figure \ref{fig:var_temp}. An exhaustive test of all edges shows that T and F are optimal moves for both players.

To go further, we first need to spell out four cases describing what happens immediately after both players follow the pattern of playing first on T and F edges.

\begin{itemize}
    \item If there are an odd number of variables and Left goes first, then after all T and F edges are gone, the remnants of those variable gadgets will have value 1 and it will be Right's turn to play.
    \item If there are an odd number of variables and Right goes first, then afterwards the remnants of those variable gadgets will have value $-1$ and it will be Left's turn.
    \item If there are an even number of variables and Left goes first, then afterwards the remnants of those variable gadgets will have value 0 and it will be Left's turn.
    \item If there are an event number of variables and Right goes first, then afterwards the remnants will have value 0 and it will be Right's turn.
\end{itemize}

Now assume the claim is false and consider a minimal counter-example with regard to the number of clause gadgets and unclaimed variables. For every possible option we demonstrate that the outcome is no better than T or F for the player that chooses it. Consider figure \ref{fig:var_temp} as a subgraph of this minimal counter-example. First, consider Left's options. Throughout, we denote by $H$ the position that consists of all clause gadgets and $s$ once all variables have been claimed. By lemma \ref{lem:lessThanNegOne} the value of $H$ is less than $-1$ if it represents a false cnf, and $-1$ if it represents a true cnf.

\begin{itemize}
    \item[$f$]: Right responds with F leaving a smaller game subject to the inductive hypothesis (wherein optimal play is in the unclaimed variable gadgets), along with an extra $-1*$. When all variables have been claimed then what remains, if there are an even number of remaining variable gadgets, is $-1*$ along with $-1$ (if all clauses are true) or strictly less than $-1$ (if at least one is false by lemma \ref{lem:lessThanNegOne}) totaling at most $-2*$. If instead an odd number of variable gadgets remain then, after they are played, what remains is $-1* + 1$ along with a position with value at most $-1$. This sums to $-1*$. 

    \item[$e$]: Right responds with $d_j$ to leave $\{1|0\}$ and a smaller position subject to the inductive hypothesis. After all variables are claimed, then what remains is $\{1|0\}$ along with $H$, totaling $\{0|-1\}$, if the cnf is true and less if the cnf is false (when an even number of variable gadgets remain) or $\{1|0\} + 1$ along with $H$ (when an odd number remain). In the former case, it's Left's turn and they can win if and only if the cnf is true. This can only occur if Left would otherwise have created a true cnf by playing T instead of $e$, granting Left no advantage by playing $e$ over $T$. Additionally, the play on edge $e$ also treats this variable as false for all other connected clauses, so Left has a chance to lose playing $e$ whereby they may have won playing $T$. In the latter case, it's Right's turn and the sum is at most $\{1|-1\}$, which Right wins. So Left has no motivation to play in $e$ instead of $T$.

    \item[$b$]: Right responds with F to leave $-2$ and a smaller position subject to the inductive hypothesis. If there are an even number of remaining variables then Left moves on $H - 2$ and loses. If there are an odd number remaining then Right moves on $H - 1$ and wins. In either case, Left loses playing this edge.

    \item[$c$]: Again, Right responds with F to $-1$ and a smaller position with an extra red edge possibly attached (labeled $a$ in the figure). This is at least as good for Right as if the additional red edge was not present. If there are an even number of variables remaining, then Left plays on $H-1$ and loses. If there are an odd number, then Right plays on $H$ and wins. 
\end{itemize}

Next we consider Right's possible moves.

\begin{itemize}
    \item[$g$]: By lemma \ref{lem:deathStar} the position that results from Right playing edge $g$ is greater than or equal to one less than the position that results from Right playing edge F instead. Since a play on $g$ leaves a value of $0$ in the variable gadget instead of the $-1$ that results from playing on F, Right does not benefit from this move.

    \item[$a$]: Left responds with $c$ to a position with fewer clauses, and no longer needs to consider this clause ending up False. Thus the outcome is at least as good for Left as otherwise. 

    \item[$d$]: Left responds with $b$ and the outcome is as in the previous case.

    \item[$s_i$]: Denote the position resulting from this move by $G$ and the position resulting from Right playing F instead by $G'$. Similarly, the red edge $a$ in $G$ has a blue analogue in $-G'$ denoted $a'$, etc. We demonstrate that in $G - G'$ Right loses moving first, and thus $G \geq G'$. To do so, we consider a Left response to every Right move. The obvious strategy is to mirror Right's $G$ moves in $-G'$ and vice versa. There are only three red edges in this sum that do not have blue analogues. 

    In response to Right playing F, Left simply plays edge $s_i'$ leaving a sum of $0$ for Right. In response to Right playing $b'$, Left plays edge $e$ to leave $1 + \{1|0\} > 0$ ($\{1|0\}$ from the variable gadget in $G$, $1$ from the disconnected variable gadget in $G'$, $-1$ from $d$, and its partner $1$ from $d'$). Finally, if Right plays $c'$, then Left responds with T in $G$. What remains is $1$ from the variable gadget in $G$, $-1$ from $d$, $1$ from the variable gadget in $G'$, and the remainder of $G$ and $G'$, each of which contains an analogous blue edge for every red edge in its mirror, along with additional blue edges that do not have red analogues. Thus, the sum is at least $1$ and a win for Left.

    Since Right loses moving first in $G - G'$, we know that $G$ is no better for Right than $G'$. Hence, Right has no motivation to play edge $s_i$ over playing edge F.        
\end{itemize}

Therefore, no option is better for either player than T or F, as appropriate, and the minimal counter-example does not exist. Hence, players are always motivated to play on variable gadgets until they are all claimed. 
\end{proof}

\begin{theorem}[Main]
For any \poscnf{} position, $G$, Left wins going first on the \parck{} position, $H$, that results from the reduction described by variable and clause gadgets exactly when True wins going first on $G$.
\label{thm:main}
\end{theorem}

\begin{proof}
    
Players will choose to play on the variable gadgets before any other parts of the graph due to lemma \ref{lem:vars_first}.  As mentioned in lemma \ref{lem:vars_first}, we assume there are an odd number of these gadgets. This will leave a total value of $1$ across all of them because one more will have been played by Left than by Right.
    
After all variables are played, those variables where Right played on the F edge will have lost their $xy$ edges to the clauses.  If any clause is missing all of these edges, say clause $j$, then that corresponds to False winning the \poscnf{} position.  Right can choose edge $s_j$ to leave that extra value of -2 in that clause and the rest at $0$ by lemma \ref{lem:noS}.  Right will win after this move because the total value of the position will be -1.  

If all clauses are still attached to at least one variable, then it corresponds exactly to True winning the \poscnf{} position.  By theorem \ref{thm:allPositive}, and the fact that we assume the cnf contains more than one clause, the value of the connected clauses is -1.  This leaves a total value of $0$ and Left has gone last, so they have a winning strategy.

\end{proof}

\begin{corollary}
    \parck{} is \cclass{PSPACE}-complete.
\end{corollary}

\begin{proof}
By \cite{DBLP:journals/jcss/Schaefer78}, \poscnf{} is \cclass{PSPACE}-complete.  Thus, \parck{} is \cclass{PSPACE}-hard due to the reduction from \poscnf{} in  theorem \ref{thm:main} .  It is also in \cclass{PSPACE} because the game tree has a bounded depth of the number of edges, $m$, and each position has at most $m$ options.
\end{proof}

\section{Conclusions and Future Work}

In this paper we showed that \parck{} is \cclass{PSPACE}-complete, even on positions without any green edges, improving on the \cclass{NP}-hardness result of \cite{HANAKA2026103716}.  \parck{} is closely related to two complexity questions that have long evaded ACGT researchers.

\begin{problem}
    What is the computational complexity of \ruleset{Domineering}?
\end{problem}

\begin{problem}
    What is the computational complexity of (impartial) \ruleset{Arc Kayles}?
\end{problem}

\ruleset{Arc Kayles} is the same as \parck{}, except that all edges are green.  

Plenty of effort has been spent on these two rulesets in the past few decades, but neither has yet been solved.

\bibliographystyle{plainurl}

\end{document}